\documentclass[reqno,a4paper, 10pt]{amsart}  % Comment this line out
\usepackage{graphicx} % for pdf, bitmapped graphics files
\usepackage{mathptmx} % assumes new font selection scheme installed
\usepackage{times} % assumes new font selection scheme installed
\usepackage{amsmath} % assumes amsmath package installed
\usepackage{amssymb}  % assumes amsmath package installed
\usepackage{datetime}
\usepackage{hyperref}
\newtheorem{thm}{Theorem}

\def\sgn{\mathrm{sgn}}

\DeclareMathAlphabet{\bit}{OML}{cmm}{b}{it}
           
\def\<{\leqslant}           % nice less than or equal to sign
\def\>{\geqslant}           % nice larger than or equal to sign
\def\d{\partial}
\def\wh{\widehat}
\def\wt{\widetilde}

\def\Re{\mathrm{Re} }   % real part
\def\mR{{\mathbb R}}    
\def\mC{\mathbb{C}}    

\def\Tr{\mathrm{Tr}}   
\def\rT{\mathrm{T}}    
\def\rF{\mathrm{F}}

\def\bE{\mathbf{E}}

\def\[[[{[\![\![}
\def\]]]{]\!]\!]}

\def\bra{\langle}
\def\ket{\rangle}

\def\re{\mathrm{e}}        
\def\rd{\mathrm{d}}        

\def\bJ{\mathbf{J}}

\def\br{\mathbf{r}}
\def\x{\times}
\def\ox{\otimes}

\def\fF{\mathfrak{F}}

\def\fH{\mathfrak{H}}

\def\sM{\mathsf{M}}
\def\sR{\mathsf{R}}

\def\sL{\mathsf{L}}

\def\sB{\mathsf{B}}

\def\eps{\epsilon}
\def\ups{\upsilon}
\def\Ups{\Upsilon}

\def\diag{\mathop\mathrm{diag}}    
\def\blockdiag{\mathop\mathrm{blockdiag}}

\numberwithin{equation}{section}

\begin{document}
\title[Decoherence as Noncommutativity Decay for Quantum Oscillators]
{Decoherence Quantification Through Commutation Relations Decay For Open Quantum Harmonic Oscillators}

%    Remove any unused author tags.

%    author two information
\author{Igor G. Vladimirov, \qquad Ian R. Petersen}
\address{School of Engineering, Australian National University, ACT 2601, Australia}
%\curraddr{School of Engineering, Australian National University, ACT 2601, Australia.   }
\email{igor.g.vladimirov@gmail.com,i.r.petersen@gmail.com}
\thanks{This work is supported by the Australian Research Council  grant DP210101938.}

%\subjclass[2000]{Primary }
%    For articles to be published after 1 January 2010, you may use
%    the following version:


\date{}

\begin{abstract}
This paper is concerned with multimode open quantum harmonic oscillators (OQHOs), described by linear quantum stochastic differential equations with multichannel external bosonic fields. We consider the exponentially fast decay in the two-point commutator matrix of the system variables as a manifestation of quantum decoherence. Such dissipative effects are  caused by the interaction of the system with its environment and lead to a loss of specific features of the unitary evolution which the system would have in the case of isolated dynamics.
These features are exploited as nonclassical resources  in quantum computation and quantum information processing technologies.
A system-theoretic  definition of decoherence time in terms of the commutator matrix decay is discussed,  and an upper bound for it is provided using algebraic Lyapunov inequalities. Employing spectrum perturbation techniques, we investigate the asymptotic behaviour of a related Lyapunov exponent for the oscillator when the system-field coupling is specified by
a small coupling strength parameter and a given coupling shape matrix.
The invariant quantum state of the system, driven by vacuum fields, in the weak-coupling limit is also studied.
We illustrate the results  for one- and two-mode oscillators with multichannel external fields and outline their application to a decoherence control problem for a feedback interconnection of OQHOs.
\end{abstract}

\maketitle

\section{Introduction}

In contrast to classical deterministic and stochastic systems, the evolution of quantum mechanical objects is described in terms of operator-valued variables on a Hilbert space,
and their statistical  properties are formulated using quantum probability \cite{H_2001,S_1994}. The quantum dynamic variables  are usually noncommuting operators, which, unlike classical variables, are  not accessible to simultaneous observation because quantum measurements, accompanied by interaction of quantum systems with classical measuring devices and resulting in real-valued quantities,  modify the quantum states. The latter, represented by density  operators on the same Hilbert space,  are quantum counterparts of classical scalar-valued probability distributions.

The more complicated nature of quantum observables and quantum states and the subtle interplay between the noncommutativity and measurements are specific features of quantum mechanics.
These features are closely related to the model of unitary, and hence reversible, dynamics of an isolated quantum system, whose initial conditions have a nonvanishing effect on the  subsequent evolution of the system. The reversibility of the isolated quantum dynamics and its consequences (such as conservation laws) play an important role in quantum computation and quantum information processing which exploit quantum mechanical resources \cite{NC_2000}.

The dynamics of an open quantum system \cite{BP_2006,GZ_2004}, which interacts with its environment (such as other quantum or classical systems and external fields), are qualitatively different from the unitary evolution. In this case, the internal dynamics, which the system would have in isolation from the surroundings, is ``diluted'' with the contribution from the other systems whose variables (as operators on different spaces) commute with (and hence, are classical with respect to) those of the underlying quantum system. This makes the open system gradually lose, over the course of time,  its original quantum features in terms of the commutation structure and statistical properties, which is interpreted as quantum decoherence.

The decoherence effects, coming from the dissipative interaction with the environment, are particularly apparent in open quantum harmonic oscillators (OQHOs) which play the role of building  blocks in linear quantum systems theory \cite{NY_2017,P_2017} aiming to engineer quantum system interconnections with desired dynamic properties, including  stability, robustness and optimality.  Such systems are endowed with conjugate position-momentum pairs \cite{S_1994} (or their linear combinations) as dynamic variables and, in the framework of the Hudson-Parthasarathy calculus \cite{HP_1984,P_1992}, are described by linear quantum stochastic differential equations (QSDEs). These equations are driven by noncommutative quantum Wiener processes on a symmetric Fock space \cite{PS_1972} which model bosonic fields. Similarly to classical linear systems \cite{KS_1972}, the solution of the linear QSDE consists of the system response to the initial system variables and the forced response to the external field, with both responses commuting  with each other. For a dissipative OQHO, with a Hurwitz dynamics matrix,  the forced response becomes dominant  over the response to the initial condition which fades away exponentially fast, thus giving rise to meaningful time scales for the quantum decoherence in the system. On the other hand, the same mechanism underlies the convergence of the stable OQHO to its invariant zero-mean Gaussian quantum state in the case of vacuum fields, and this is employed for the generation of such states through dissipation \cite{Y_2012}.

The present paper is concerned with a particular way to quantify decoherence for multimode dissipative OQHOs, which defines the decoherence time in terms of the exponentially fast decay in the two-point canonical commutation relations (CCRs) for the system variables. We provide  a system-theoretic   upper bound for the decoherence time using algebraic Lyapunov inequalities.
We also investigate the related asymptotic behaviour of the  leading Lyapunov exponent for the dynamics matrix of the oscillator with a small coupling strength parameter and a given coupling shape matrix. This analysis employs spectrum perturbation techniques \cite{M_1985} and    is carried out for a class of dissipative OQHOs with a positive definite energy matrix, a nonsingular one-point CCR matrix and pairwise different eigenfrequencies for the uncoupled version of the system. The weak-coupling formulation allows a low decoherence criterion for the OQHO to be obtained in the form of the asymptotic decoherence time significantly exceeding the largest period of uncoupled oscillatory modes.
Since lowering the decoherence towards isolating the system from its environment is in conflict with increasing the dissipation for accelerated convergence to the invariant zero-mean Gaussian quantum state mentioned above, we also discuss the weak-coupling limit for the invariant covariance matrix which completely specifies such a state.
These asymptotic results are compared with exact computations for one- and two-mode oscillators  driven by multichannel fields. We also apply them to a decoherence control setting which is outlined for a coherent (measurement-free) feedback interconnection of two OQHOs (playing the role of a quantum plant and a quantum controller) with direct and indirect coupling.

Note that the type of decoherence discussed in this paper corresponds to vacuum decoherence in \cite[Section 4.4.1]{BP_2006}, while there also are other scenarios of decoherence, involving, for example, thermal quantum noise \cite[Section~3.3.3]{GZ_2004}, quantum measurements and different classes of quantum  systems, such as those with finite-level dynamic variables (see also \cite[Section 4.4.1]{BP_2006} and references therein, including \cite{CL_1985,U_1995}). Although,  in the context of quantum computing and quantum information, the study of decoherence phenomena is particularly relevant for finite-level (multiqubit) systems \cite[Chapter III, Section 8]{NC_2000} which employ the Pauli matrices \cite{S_1994} and their extensions, OQHOs are tractable as the closest quantum-mechanical  counterparts of classical linear stochastic systems, and the results for OQHOs can serve as prototypes for a more complicated decoherence analysis in the finite-level case.

The paper is organised as follows.
Section~\ref{sec:sys} describes the class of OQHOs under consideration,  including the isolated dynamics case in the absence of coupling.
Section~\ref{sec:decay} discusses the exponential decay in the two-point CCRs for the system variables and the corresponding decoherence time.
Section~\ref{sec:asy} obtains asymptotic estimates for the Lyapunov exponents and decoherence time in the presence of a small coupling strength parameter.
Section~\ref{sec:cov} establishes a weak-coupling limit for the invariant covariance matrix.
Section~\ref{sec:one} shows that the decoherence time estimates are exact for one-mode oscillators.
Section~\ref{sec:num} provides a numerical illustration of the  results for a two-mode oscillator.
Section~\ref{sec:two} outlines an application of the results to a decoherence control setting for interconnected OQHOs.
Section~\ref{sec:conc} makes concluding remarks.

\section{Open quantum harmonic oscillators}
\label{sec:sys}

We consider an $\frac{n}{2}$-mode open quantum harmonic oscillator (OQHO) with an even number $n$ of dynamic variables $X_1, \ldots, X_n$. In accordance with the Heisenberg picture of quantum dynamics \cite{S_1994},  they are time-varying self-adjoint  operators on a complex separable Hilbert space $\fH$, satisfying at every moment of time the canonical commutation relations (CCRs)
\begin{equation}
\label{XCCR}
    [X,X^\rT] = 2i\Theta,
    \qquad
    X := \begin{bmatrix}
      X_1\\
      \vdots\\
      X_n
    \end{bmatrix}
\end{equation}
with a constant matrix $\Theta = -\Theta^\rT \in \mR^{n\x n}$. Here, the matrix transpose $(\cdot)^\rT$ applies to vectors of operators as if they consisted of scalars  (with vectors being organised as columns unless indicated otherwise), and the commutator $[\alpha, \beta]:= \alpha \beta - \beta \alpha$ of linear operators extends to the commutator matrix $[\xi, \eta^\rT] := ([\xi_j, \eta_k])_{1\< j \< a, 1\< k \< b}$ for vectors $\xi:= (\xi_j)_{1\< j \< a}$ and $\eta:= (\eta_k)_{1\< k \< b}$ of operators. The internal energy of the oscillator is described by the Hamiltonian
\begin{equation}
\label{H}
  H:= \frac{1}{2}X^\rT R X,
\end{equation}
which is a self-adjoint operator on $\fH$, whose quadratic dependence on the system variables is parameterised by an energy matrix $R = R^\rT \in \mR^{n\x n}$. As an open system, the OQHO interacts with its environment which is modelled by  an $\frac{m}{2}$-channel external bosonic field in the form of an even number $m$ of self-adjoint quantum Wiener processes $W_1, \ldots, W_m$ on a symmetric Fock space $\fF$ \cite{P_1992} with the Ito table
\begin{equation}
\label{Omega}
    \rd W \rd W^\rT
    = \Omega \rd t,
    \qquad
    W := \begin{bmatrix}
      W_1\\
      \vdots\\
      W_m
    \end{bmatrix},
    \qquad
    \Omega: = I_m + iJ,
\end{equation}
where
\begin{equation}
\label{JJ}
    J
    :=
    \bJ \ox I_{m/2}
    =
    \begin{bmatrix}
      0 & I_{m/2}\\
    -I_{m/2} & 0
    \end{bmatrix}.
\end{equation}
Here, $\ox$ is the Kronecker product of matrices, $I_r$ is the identity matrix of order $r$, and the matrix
\begin{equation}
\label{bJ}
        \bJ
        : =
        {\begin{bmatrix}
        0 & 1\\
        -1 & 0
    \end{bmatrix}}
\end{equation}
spans the one-dimensional subspace of antisymmetric matrices of order 2. In accordance with (\ref{Omega}) and the commutativity $[W, \rd W^\rT] = 0$ between $W$ (or any other adapted process) and the future-pointing Ito increments $\rd W$,  the matrix $J$ in (\ref{JJ}) specifies the commutation structure of $W$ as
\begin{equation*}
\label{WWcomm}
  [W(s), W(t)^\rT]
  =
  2i \min(s,t)J,
  \qquad
  s, t \> 0.
\end{equation*}
The internal dynamics of the system and its interaction with the external field lead to a quantum stochastic differential equation (QSDE)  \cite{NY_2017,P_2017} for the time evolution of the system variables:
\begin{equation}
\label{dX}
    \rd X  = AX \rd t + B\rd W,
\end{equation}
which is understood in the sense of the Hudson-Parthasarathy calculus \cite{HP_1984,P_1992}.
Here, $A \in \mR^{n\x n}$, $B \in \mR^{n\x m}$ are constant matrices which are specified by the CCR matrices  $\Theta$, $J$  from (\ref{XCCR}), (\ref{JJ}),  the energy matrix $R$ from (\ref{H}), and a system-field coupling matrix $M \in \mR^{m\x n}$ as
\begin{equation}
\label{AB}
    A = 2\Theta (R + M^\rT J M),
     \qquad
     B = 2\Theta M^\rT.
\end{equation}
The matrix $M$ parameterises a vector  $MX$ of $m$ self-adjoint coupling operators, which pertain to the energy exchange between the OQHO and the external quantum field.  In the absence of coupling between the system and the environment, when $M=0$, the matrix $A$ takes the form
\begin{equation}
\label{A0}
    A_0 := 2\Theta R,
\end{equation}
and the dispersion matrix $B$ in (\ref{AB}) vanishes, so that
the QSDE (\ref{dX}) loses its diffusion term $B\rd W$ and reduces to an ODE
\begin{equation}
\label{Xdot}
  \dot{X} =
  i[H,X]
  =
  A_0X,
\end{equation}
where $\dot{(\ )}:= \d_t(\, )$ is the time derivative. In this isolated dynamics case, where $X(t) = U(t)^\dagger X(0) U(t)$, with $U(t):=\re^{-itH(0)}$ a unitary operator and $(\cdot)^\dagger$ the operator adjoint, the Hamiltonian $H(t)= U(t)^\dagger H(0) U(t) = H(0)U(t)^\dagger U(t) = H(0)$  is preserved in time. This can also be seen from  (\ref{H}), (\ref{Xdot})  as
$$
    \dot{H}
    =
    \frac{1}{2} (\dot{X}^\rT R X + X^\rT R\dot{X} )
    =
    \frac{1}{2}
    X^\rT (A_0^\rT R + R A_0)X
    =
    0
$$
since, in view of (\ref{A0}),
\begin{equation}
\label{A0RRA0}
    \frac{1}{2} (A_0^\rT R + R A_0) = (\Theta R)^\rT R + R \Theta R = 0
\end{equation}
due to the antisymmetry of the CCR matrix $\Theta$ and the symmetry of the energy matrix $R$. For the general QSDE (\ref{dX}), the Hamiltonian $H$ is no longer a constant operator. More precisely,  by applying the quantum Ito lemma \cite{HP_1984,P_1992} along with (\ref{Omega}) to (\ref{H}), it follows that
\begin{align}
\nonumber
    \rd H
    &=
    \frac{1}{2}
    ((\rd X)^\rT R X + X^\rT R \rd X + (\rd X)^\rT R \rd X) \\
\nonumber
    & =
    \frac{1}{2}
    (X^\rT(A^\rT R + RA) X\rd t + 2X^\rT R B\rd W + (\rd W)^\rT B^\rT R B \rd W) \\
\nonumber
    & =
    \frac{1}{2}
    X^\rT(\wt{A}^\rT R + R\wt{A}) X\rd t
    +
    X^\rT R B\rd W +
    \frac{1}{2}
    \bra
        B^\rT R B,
        \Omega
    \ket_\rF\rd t\\
\label{dH}
    & =
    \frac{1}{2}(X^\rT(\wt{A}^\rT R + R\wt{A}) X +\bra R, BB^\rT\ket_\rF)\rd t
    +
    X^\rT R B\rd W ,
\end{align}
where
\begin{equation}
\label{At}
  \wt{A} := A-A_0 = 2\Theta M^\rT J M,
  \qquad
  BB^\rT= -4\Theta M^\rT M \Theta,
\end{equation}
in accordance with (\ref{AB}), (\ref{A0}).
Here, $\bra \cdot, \cdot \ket_\rF$ is the Frobenius inner product of matrices \cite{HJ_2007}, and  use is also made of the commutativity $[X, \rd W^\rT] = 0$ together with the special structure of the quantum Ito matrix matrix $\Omega$ from (\ref{Omega}) and the matrix $A$ in (\ref{AB}), with the latter implying that
\begin{equation}
\label{ARRA}
  A^\rT R + RA
  =
  \wt{A}^\rT R + R\wt{A}
\end{equation}
in view of (\ref{A0RRA0}), (\ref{At}).
In what follows,  the expectation
\begin{equation}
\label{bE}
    \bE \zeta := \Tr(\rho \zeta)
\end{equation}
of a quantum variable $\zeta$ on the system-field space $\fH:= \fH_0 \ox \fF$ is over the tensor-product state
$$    
    \rho = \rho_0 \ox \ups,
$$
where $\rho_0$ is the initial quantum state of the system on the initial system space $\fH_0$ (for the action of $X_1(0), \ldots, X_n(0)$), and $\ups$ is the vacuum state \cite{P_1992} for the quantum Wiener process $W$ on $\fF$. In the case of  vacuum fields,  the averaging of both sides of (\ref{dH}) according to (\ref{bE}) yields
\begin{align}
\nonumber
    (\bE H)^{^\centerdot}
    & =
    \frac{1}{2}
    (\bE(X^\rT
    (\wt{A}^\rT R + R\wt{A})
    X)
    +
    \bra R, BB^\rT\ket_\rF)\\
\label{EHdot}
    & =
    \frac{1}{2}
    (\bra
        \wt{A}^\rT R + R\wt{A},
        P
    \ket_\rF
    +
    \bra R, BB^\rT\ket_\rF),
\end{align}
where the martingale part  $X^\rT R B\rd W$ of (\ref{dH}) does not contribute to the expectation, and use is made of the real covariances
\begin{equation*}
\label{P}
  P:= \Re \bE (XX^\rT)
\end{equation*}
of the system variables. The time-varying matrix $P = P^\rT \in \mR^{n\x n}$ satisfies
\begin{equation}
\label{PT}
    P+i\Theta = \bE(XX^\rT) \succcurlyeq 0
\end{equation}
due to the generalised Heisenberg uncertainty principle \cite{H_2001} and is governed by the Lyapunov ODE
\begin{equation}
\label{Pdot}
    \dot{P} = AP + PA^\rT + BB^\rT.
\end{equation}
The latter provides an alternative way to arrive at (\ref{EHdot}) through the relation $\bE H = \frac{1}{2} \bra R, P\ket_\rF$ as
\begin{align}
\nonumber
    (\bE H)^{^\centerdot}
    & =
    \frac{1}{2} \bra R, \dot{P}\ket_\rF\\
\nonumber
    & =
    \frac{1}{2} \bra R, AP + PA^\rT + BB^\rT\ket_\rF\\
\nonumber
    & =
    \frac{1}{2} (\bra A^\rT R + RA,P\ket_\rF + \bra R, BB^\rT\ket_\rF)\\
\label{EHdot1}
    & =
    \frac{1}{2} (\bra \wt{A}^\rT R + R\wt{A},P\ket_\rF + \bra R, BB^\rT\ket_\rF),
\end{align}
where (\ref{Pdot}) is used along with (\ref{ARRA}) and the second equality from (\ref{At}). As seen from the right-hand side of (\ref{EHdot}),   for a given matrix $P$, the quantity $(\bE H)^{^\centerdot}$, which pertains to the rate of work of the external field on the system, depends on the coupling matrix $M$ in a quadratic fashion through the matrices (\ref{At}). In the steady-state regime, when the matrix $A$ is Hurwitz and $P$ is the unique solution
\begin{equation}
\label{P}
    P = \int_{\mR_+} \re^{tA }BB^\rT \re^{tA^\rT} \rd t
\end{equation}
of the algebraic Lyapunov equation (ALE)
\begin{equation}
\label{PALE}
    A P + P A^\rT + BB^\rT = 0,
\end{equation}
it follows from (\ref{EHdot1}) that $(\bE H)^{^\centerdot} = 0$, and hence,
the Hamiltonian $H$, despite not being preserved, has a constant mean value $\bE H$. In this case  (when $A$ is Hurwitz),
$    P+i\Theta = \bE(XX^\rT)
$
is the quantum covariance matrix of the invariant zero mean Gaussian quantum state \cite{KRP_2010} of the system.
This property of OQHOs underlies the generation of Gaussian states through a sufficiently long run of the system \cite{Y_2012}. The dissipation, which is exploited in this  state generation procedure,  comes with decoherence. The latter  results from the system-environment interaction and leads to a ``loss of quantumness'' over time in comparison with the isolated dynamics.

\section{Exponential decay in two-point CCRs}
\label{sec:decay}

The Hurwitz property of the matrix $A$, which secures the steady-state regime discussed above, gives rise to an exponentially fast decay in the
two-point CCRs \cite{VPJ_2018a} for the system variables described by
\begin{equation}
\label{XXCCR}
    [X(s), X(t)^\rT]
    =
    2i\Ups(s-t),
    \qquad
    s,t\>0,
\end{equation}
with
\begin{equation}
\label{Lambda}
    \Ups(\tau)
     :=
    \left\{
    {\small\begin{matrix}
    \re^{\tau A}\Theta & {\rm if}\  \tau\> 0\\
    \Theta\re^{-\tau A^{\rT}} & {\rm if}\  \tau< 0\\
    \end{matrix}}
    \right.,
    \qquad
    \tau \in\mR,
\end{equation}
from which the one-point CCRs (\ref{XCCR}) are obtained as a particular case  at $s=t$ (that is, $\tau=0$) due to
\begin{equation}
\label{Lambda0}
    \Ups(0)= \Theta.
\end{equation}
The CCRs (\ref{XXCCR}), which  are a consequence of the joint commutation structure of the system variables and  the external fields, hold regardless of a particular quantum state and regardless of $A$ being Hurwitz. The limit relation
\begin{equation}
\label{Lam0}
    \lim_{\tau \to \infty}\Ups(\tau) = 0
\end{equation}
(when $A$ is Hurwitz) means that for any $j,k=1, \ldots, n$, the system variables $X_j(s)$ become ``asymptotically commuting'' with (and hence, classical with respect to) $X_k(t)$  at distant moments of time $s$, $t$  (as $|s-t|\to +\infty$). The decay (\ref{Lam0})  in the two-point CCR matrix can be regarded as a manifestation of quantum decoherence in the system caused by its coupling to the environment. Indeed, the forced response term $\int_0^t \re^{(t-s)A}B\rd W(s)$, contributed by the external field $W$ to the system variables
\begin{equation}
\label{XW}
    X(t) = \re^{tA} X(0) + \int_0^t \re^{(t-s)A}B\rd W(s)
\end{equation}
over time $t\> 0$,
and the response $\re^{tA} X(0)$ of the system to the initial condition commute with each other:
\begin{equation}
\label{comm}
    \Big[
        \int_0^t \re^{(t-s)A}B\rd W(s),
        (\re^{tA} X(0))^{\rT}
    \Big] =
    \int_0^t
    \re^{(t-s)A}B
    [\rd W(s), X(0)^\rT]
    \re^{tA^\rT}
    =
    0
\end{equation}
since $X(0)$ and $W$ consist of operators acting on different Hilbert spaces $\fH_0$ and $\fF$, whereby $[W(t), X(0)^\rT] = 0$ for all $t\> 0$. In fact, (\ref{XXCCR}), (\ref{Lambda}) follow from a combination of (\ref{XW}), (\ref{comm}) with (\ref{XCCR}). Therefore, if $A$ is Hurwitz, $X(0)$ has an exponentially decaying effect on $X(t)$, and the forced term in (\ref{XW}) becomes dominant over the course of time. This dissipative behaviour is qualitatively different from the isolated dynamics in the absence of coupling ($M=0$), when the matrix $A$ in (\ref{AB}) reduces to $A_0$ in (\ref{A0}),  which
has a purely imaginary spectrum $\{i\omega_k: k = 1, \ldots, n\}$  in the case of a nonsingular CCR matrix $\Theta$ and positive definite energy matrix $R$:
\begin{equation}
\label{typ}
    \det \Theta \ne 0,
    \qquad
    R \succ 0 .
\end{equation}
The eigenfrequencies $\omega_1, \ldots, \omega_n \in \mR$ (which should not be confused with the spectrum of the Hamiltonian $H$ in (\ref{H}) as a self-adjoint operator on the infinite dimensional system-field Hilbert space $\fH$) are all nonzero and symmetric about the origin, so that, without loss of generality,
\begin{equation}
\label{sym}
    \omega_k=-\omega_{k+\frac{n}{2}}>0,
  \qquad
  k = 1, \ldots, \frac{n}{2}.
\end{equation}
Indeed, under the conditions (\ref{typ}), the matrix  $A_0$ in (\ref{A0}) is isospectral to the nonsingular real antisymmetric matrix $2\sqrt{R}\Theta \sqrt{R}$ whose eigenvalues are purely imaginary and symmetric about the origin \cite{HJ_2007}. This isospectrality follows from the similarity transformation
\begin{equation}
\label{AR}
    A_0 = R^{-1/2} (2\sqrt{R}\Theta \sqrt{R}) \sqrt{R}
\end{equation}
(which is used, for example, in \cite{P_2014}; see also the proof of Williamson's symplectic diagonalization  theorem \cite{W_1936,W_1937} in \cite[pp. 244--245]{D_2006}). Moreover, (\ref{AR}) implies that the matrix
$A_0$ is diagonalisable as
\begin{equation}
\label{AV}
    A_0 = i S \mho S^{-1},
    \qquad
    S := R^{-1/2}V,
    \qquad
    \mho := \diag_{1\< k\< n} (\omega_k).
\end{equation}
Here, $S^{-1} = V^* \sqrt{R}$, with $(\cdot)^*:= {\overline{(\cdot)}}^\rT$ the complex conjugate transpose of a matrix,    and
\begin{equation}
\label{V}
    V
    :=
    \begin{bmatrix}
        v_1 & \ldots & v_n
    \end{bmatrix}
    \in \mC^{n\x n}
\end{equation}
is a unitary matrix
whose columns $v_1,\ldots, v_n\in \mC^n$ are the eigenvectors of the Hermitian matrix \begin{equation}
\label{VV}
    -2i\sqrt{R}\Theta \sqrt{R} = V\mho V^*
\end{equation}
with the eigenvalues $\omega_1, \ldots, \omega_n$,
\begin{equation}
\label{veig}
    -2i\sqrt{R}\Theta \sqrt{R} v_k = \omega_k v_k,
    \qquad
    k = 1, \ldots, n,
\end{equation}
satisfying
\begin{equation}
\label{vsym}
    v_{k+\frac{n}{2}} = \overline{v_k},
    \qquad
    1\< k \< \frac{n}{2} ,
\end{equation}
in accordance with (\ref{sym}). 
The corresponding matrix exponentials
$$
    \re^{\tau A_0}
    =
    S\re^{i\tau  \mho} S^{-1},
    \qquad
    \re^{-\tau A_0^\rT}
    =
    S^{-\rT}
    \re^{-i\tau  \mho}
    S^\rT
$$
(with $(\cdot)^{-\rT}:= ((\cdot)^{-1})^\rT$)
in (\ref{Lambda}), associated with (\ref{AV}),  are oscillatory functions of time $\tau$, with
\begin{equation}
\label{T*}
  T :=
  \frac{2\pi}{\min_{1 \< k \< \frac{n}{2}} \omega_k} >0
\end{equation}
being the largest period of the oscillatory modes, since $\re^{i \tau \mho} = \diag_{1\< k\< n} (\re^{i\omega_k\tau})$.
Therefore, in the case of (\ref{typ}), only a nonzero coupling matrix $M$ can make the matrix $A = A_0+\wt{A}$ in (\ref{AB}) Hurwitz (thus leading to (\ref{Lam0})) through the matrix $\wt{A}$ in (\ref{At}) 
in view of (\ref{A0}). If the matrix $A$ is Hurwitz, a ``typical'' time constant of the exponential decay in the two-point CCR matrix (\ref{Lambda}) provides a measure of quantum decoherence in the system and can be defined, for example,  as
\begin{equation}
\label{taud}
    \tau_*
    :=
    \inf
    \Big\{
        \tau>0:\
        \|\re^{\tau A}\Theta\|_\rF \< \frac{1}{\re}\|\Theta\|_\rF
    \Big\}.
\end{equation}
Here, use is made of the Frobenius norm  $\|\Theta\|_\rF = \sqrt{-\Tr (\Theta^2)}$ of the one-point antisymmetric CCR matrix  $\Theta$ from  (\ref{Lambda0}), although a different matrix norm can also be used for this purpose.

An upper bound for the decoherence time (\ref{taud}) can be obtained through algebraic Lyapunov inequalities.
More precisely, for any
\begin{equation}
\label{mu1}
    0< \lambda < -\ln \br(\re^A) = -\max_{1\< k \< n}\Re \lambda_k
\end{equation}
(with $\br(\cdot)$ the spectral radius, and $\lambda_1, \ldots, \lambda_n$ the eigenvalues of $A$), there exists a positive definite matrix $\Gamma = \Gamma^\rT \in \mR^{n\x n}$ such that
\begin{equation}
\label{mu2}
    A\Gamma + \Gamma A^\rT + 2\lambda \Gamma
    \prec 0.
\end{equation}
Since the second inequality in (\ref{mu1}) is equivalent to $A+\lambda I_n$ being Hurwitz, all the matrices $\Gamma$,  satisfying (\ref{mu2}), can be represented as the solutions
\begin{equation}
\label{N}
    \Gamma
    =
    \int_{\mR_+}
    \re^{2\lambda t}
    \re^{tA}
    N
    \re^{tA^\rT}
    \rd t
\end{equation}
of the ALEs
$$
    (A+\lambda I_n)\Gamma + \Gamma (A+\lambda I_n)^\rT + N = 0
$$
(by analogy with (\ref{P}), (\ref{PALE})),
involving arbitrary positive definite matrices $N = N^\rT \in \mR^{n\x n}$.
Similarly to \cite[Proof of Theorem 6 on p. 122]{VPJ_2018a},  (\ref{mu2}) implies the contraction property
$$
    \|G_\tau \| \< \re^{-\lambda\tau},
    \qquad
    \tau> 0
$$
for the matrix
$$
    G_\tau:=
    \Gamma^{-1/2}\re^{\tau A}\sqrt{\Gamma}
$$
in the sense of the operator matrix norm $\|\cdot\|$.
Hence,
\begin{align*}
    \|\re^{\tau A}\Theta\|_\rF^2
    & =
    \|\sqrt{\Gamma}  G_\tau \Gamma^{-1/2}\Theta\|_\rF^2\\
    & =
    \Tr ((\Gamma^{-1/2}\Theta )^\rT
    G_\tau^\rT
    \Gamma
    G_\tau
    \Gamma^{-1/2}\Theta)\\
    &
    \<
    \lambda_{\max}(\Gamma)
    \Tr ((\Gamma^{-1/2}\Theta )^\rT
    G_\tau^\rT
    G_\tau
    \Gamma^{-1/2}\Theta)\\
    &
    \<
    \|\Gamma\|
        \|G_\tau\|^2
    \Tr ((\Gamma^{-1/2}\Theta )^\rT
    \Gamma^{-1/2}\Theta)\\
    & \<
    \re^{-2\lambda\tau}
    \|\Gamma\|
    \|\Gamma^{-1/2}\Theta\|_\rF^2
\end{align*}
(where $\lambda_{\max}(\cdot)$ is the largest eigenvalue of a matrix with a real spectrum),
so that
$    \|\re^{\tau A}\Theta\|_\rF \<     \re^{-\lambda\tau}
    \sqrt{\|\Gamma\|}
    \|\Gamma^{-1/2}\Theta\|_\rF
$,
and  the decoherence time (\ref{taud}) admits an upper bound
\begin{equation}
\label{tau*max}
    \tau_*
    \<
    \frac{1}{\lambda}
    \Big(
        1
        +
        \ln
        \Big(
            \frac{1}
            {\|\Theta\|_\rF}
            \sqrt{\|\Gamma\|}\|\Gamma^{-1/2}\Theta\|_\rF
        \Big)
    \Big).
\end{equation}
In view of the parameterisation (\ref{N}) for the matrix $\Gamma$, the bound (\ref{tau*max}) can be tightened by minimising its right-hand side over the pairs $(\lambda,N)$ such that $\lambda$  satisfies (\ref{mu1}) and $N\succ 0$ is normalised, for example, as $\Tr N = 1$. The normalisation does not affect the resulting minimum value due to the linearity of the map $N \mapsto \Gamma$ and the invariance of the quantity $\sqrt{\|\Gamma\|}\|\Gamma^{-1/2}\Theta\|_\rF$ in (\ref{tau*max}) with respect to the scaling transformation $\Gamma\mapsto \sigma \Gamma$ for any $\sigma>0$.

\section{Asymptotic weak-coupling decoherence estimates}
\label{sec:asy}

With the decoherence part $\wt{A}$ 
of the matrix  $A$ in (\ref{At}) depending on $M$ in a quadratic fashion,  this homogeneity can be taken into account by considering a weak-coupling formulation
\begin{equation}
\label{Meps}
    M_\eps := \eps \sM,
\end{equation}
where $\eps \> 0$ is a small scaling factor which quantifies the coupling strength, while $\sM \in \mR^{m\x n}$ specifies the coupling ``shape''. Accordingly, the matrices (\ref{AB}) acquire dependence on $\eps$ as
\begin{equation}
\label{ABeps}
    A_\eps
    =
    A_0 + \wt{A}_\eps,
    \quad
    \wt{A}_\eps
    =
    \eps^2 \wt{A}_1,
    \quad
    \wt{A}_1:= 2\Theta \sM^\rT J \sM,
     \quad
     B_\eps
     =
     \eps \sB,
     \quad
     \sB := 2\Theta \sM^\rT,
\end{equation}
where $A_0$ is given by (\ref{A0}). 
The following theorem is concerned with the asymptotic behaviour of the spectrum of the dynamics matrix $A_\eps$.

\begin{thm}
\label{th:asy}
Suppose the CCR matrix $\Theta$ and the energy matrix $R$ of the OQHO satisfy (\ref{typ}). Also, suppose the eigenfrequencies of the matrix $A_0$ from (\ref{A0})  are   pairwise different:
\begin{equation}
\label{omdiff}
  \omega_j \ne \omega_k,
  \qquad
  1\< j\ne  k \< n.
\end{equation}
Then for all sufficiently small values of the coupling strength parameter $\eps$, the matrix $A_\eps$ in (\ref{ABeps}) has different eigenvalues $\lambda_1(\eps), \ldots, \lambda_n(\eps)$ which, appropriately numbered, behave asymptotically as
\begin{equation}
\label{eigasy}
  \lambda_k(\eps)
  =
  \omega_k(i - \eps^2\mu_k) + o(\eps^2),
  \qquad
  {\rm as}\
  \eps \to 0+.
\end{equation}
Here,
\begin{equation}
\label{mu}
  \mu_k:= -i v_k^* R^{-1/2} \sM^\rT J \sM R^{-1/2} v_k,
  \qquad
  k = 1, \ldots, n,
\end{equation}
are real-valued quantities involving
the orthonormal eigenvectors $v_k$ from (\ref{veig}) and the coupling shape matrix $\sM$ in (\ref{Meps}). \hfill$\square$
\end{thm}
\begin{proof}
The property that the eigenvalues $\lambda_1(\eps), \ldots, \lambda_n(\eps)$ of the matrix $A_\eps$ are pairwise different for all sufficiently small $\eps$ (that is, $\eps \in [0, \delta)$ for some $\delta>0$) follows from (\ref{omdiff}) and  the Gershgorin localisation theorem \cite{H_2008,HJ_2007}. This also allows them to be numbered in such a way that each eigenvalue  $\lambda_k(\eps)$ inherits the infinite differentiability from $A_\eps$ in (\ref{ABeps}) over  $\eps \in [0, \delta)$. In view of the similarity transformation in (\ref{AV}),  the matrix $A_\eps$ is isospectral to
\begin{align}
\nonumber
    S^{-1}A_\eps S
    & =
    i \mho  + 2\eps^2 S^{-1}\Theta \sM^\rT J \sM S \\
\nonumber
    & =
    i \mho  + 2\eps^2 V^*\sqrt{R}\Theta \sM^\rT J \sM R^{-1/2}V\\
\nonumber
    & =
    i \mho  + 2\eps^2 V^*\sqrt{R}\Theta\sqrt{R}R^{-1/2} \sM^\rT J \sM R^{-1/2}V\\
\nonumber
    & =
    i \mho  + i\eps^2 V^*V\mho V^*R^{-1/2} \sM^\rT J \sM R^{-1/2}V\\
\label{SS}
    & =
    i \mho  + i\eps^2 \mho V^*R^{-1/2} \sM^\rT J \sM R^{-1/2}V,
\end{align}
where use is also made of (\ref{V}), (\ref{VV}). Since the matrix $\mho$ in (\ref{AV}) is diagonal, and its diagonal entries $\omega_k$ are pairwise different in view of (\ref{omdiff}), then by applying the spectrum perturbation results (see, for example, \cite{M_1985} and references therein), it follows from (\ref{SS}) that the diagonal entries
\begin{equation}
\label{kk}
    (i\mho V^*R^{-1/2} \sM^\rT J \sM R^{-1/2}V)_{kk}
    =
    i\omega_k v_k^*R^{-1/2} \sM^\rT J \sM R^{-1/2}v_k
    =
    -\omega_k \mu_k
\end{equation}
specify the coefficients of the linear terms in the asymptotic expansions of $\lambda_k(\eps)$ over the powers of $\eps^2$ as described in (\ref{eigasy}). Here, the matrix $i V^* R^{-1/2} \sM^\rT J \sM R^{-1/2}V$ is Hermitian, and hence, its diagonal entries  $i v_k^*R^{-1/2} \sM^\rT J \sM R^{-1/2}v_k$ are real, and so also are the quantities $\mu_k$ in  (\ref{mu}).
\end{proof}

The following theorem is a corollary of Theorem~\ref{th:asy} and provides stability conditions  for the weakly coupled system.

\begin{thm}
\label{th:stab}
Suppose the assumptions of Theorem~\ref{th:asy} are satisfied and the eigenfrequencies $\omega_k$ of the matrix $A_0$ in  (\ref{A0}) are arranged according to (\ref{sym}). Then the fulfillment of the inequalities
\begin{equation}
\label{mupos}
  \mu_k>0,
  \qquad
  k=1, \ldots, \frac{n}{2},
\end{equation}
is sufficient (and in nonstrict form, necessary) for the matrix $A_\eps$ in (\ref{ABeps}) to be  Hurwitz for all $\eps> 0$ small enough. \hfill$\square$
\end{thm}
\begin{proof}
From (\ref{eigasy}), it follows that
\begin{equation}
\label{Reeigasy}
  \Re \lambda_k(\eps)
  =
  - \omega_k\mu_k \eps^2 + o(\eps^2),
  \qquad
  {\rm as}\
  \eps \to 0+.
\end{equation}
Hence, the fulfillment of the inequalities
\begin{equation}
\label{ommu1}
    \omega_k \mu_k>0,
    \qquad
    k=1, \ldots, n,
\end{equation}
is sufficient for $A_\eps$ to be Hurwitz for all $\eps> 0$ small enough, while the nonstrict inequalities
\begin{equation}
\label{ommu2}
    \omega_k\mu_k = -\lim_{\eps\to 0} \frac{\Re \lambda_k(\eps)}{\eps^2}\> 0,
    \qquad
    k=1, \ldots, n,
\end{equation}
provide a necessary condition for the Hurwitz property. With the eigenvectors $v_k$ of the matrix (\ref{VV}) satisfying (\ref{vsym}) due to the convention (\ref{sym}), it follows for the quantities (\ref{mu}) that
\begin{align*}
  \mu_{k+\frac{n}{2}}
  & = -i (\overline{v_k})^* R^{-1/2} \sM^\rT J \sM R^{-1/2} \overline{v_k}\\
  & = \overline{i v_k^* R^{-1/2} \sM^\rT J \sM R^{-1/2} v_k}
  =-\mu_k,
\end{align*}
whereby
\begin{equation}
\label{ommu3}
    \omega_{k+\frac{n}{2}}\mu_{k+\frac{n}{2}} = \omega_k\mu_k,
    \qquad
  k = 1, \ldots, \frac{n}{2}.
\end{equation}
Since $\sgn(\omega_k\mu_k) = \sgn \mu_k$ for all $  k = 1, \ldots, \frac{n}{2}$ in view of (\ref{sym}), then (\ref{ommu3}) reduces the inequalities (\ref{ommu1}) to (\ref{mupos}),  and a similar reduction holds for their nonstrict versions in  (\ref{ommu2}).
\end{proof}

Under the conditions of Theorem~\ref{th:stab}, it follows from a combination of (\ref{Reeigasy}) with (\ref{ommu3}) that
\begin{equation}
\label{lead}
    \ln \br(\re^{A_\eps})
    =
    \max_{1\< k \< n}
    \Re \lambda_k(\eps)
    =
    - \eps^2 \min_{1\< k \< \frac{n}{2}}(\omega_k\mu_k) + o(\eps^2),
  \qquad
  {\rm as}\
  \eps \to 0+.
\end{equation}
The leading Lyapunov exponent term $\eps^2 \min_{1\< k \< \frac{n}{2}}(\omega_k\mu_k)$ (taken with the opposite sign) provides an asymptotically accurate approximation for the right-hand side of (\ref{mu1}) and also an estimate
\begin{equation}
\label{tauhat}
    \wh{\tau}
    :=
    \frac{\eps^{-2}}
    {\min_{1\< k \< \frac{n}{2}}(\omega_k\mu_k)}
\end{equation}
for a decay time, which is different from yet related to (\ref{taud}). The decoherence time estimate (\ref{tauhat}) allows the coupling strength parameter $\eps$ to be chosen so as to enable all the oscillatory modes (of the uncoupled oscillator)  to manifest themselves in a large number of cycles before the decay sets in. This requirement takes the form $\wh{\tau}\gg T$ in terms of the period (\ref{T*}) and leads to
\begin{equation}
\label{epsmax}
    \eps
    \ll
    \sqrt{
    \frac
    {\min_{1 \< k \< \frac{n}{2}} \omega_k}
    {2\pi\min_{1\< k \< \frac{n}{2}}(\omega_k\mu_k)}}
    =:
    \wh{\eps}
\end{equation}
as an asymptotic threshold on the coupling strength for the OQHO to preserve the quantum dynamic features of its isolated counterpart. Since
\begin{equation}
\label{mumax}
    \wh{\eps}
    \>
    \frac{1}{\sqrt{2\pi\max_{1\< k \< \frac{n}{2}}\mu_k}}
    =:
    \wt{\eps},
\end{equation}
in view of
$$
    \min_{1\< k \< \frac{n}{2}}(\omega_k\mu_k)
    \<
    \min_{1\< k \< \frac{n}{2}}\omega_k
    \max_{1\< k \< \frac{n}{2}}\mu_k     ,
$$
the right-hand side of (\ref{mumax}) provides a more stringent threshold on $\eps$. For a one-mode oscillator ($n=2$)  with $\mu_1>0$, the eigenfrequency $\omega_1>0$ cancels out in (\ref{epsmax}), so that  both thresholds 
are the same and reduce to
\begin{equation}
\label{eps<<}
    \eps
    \ll
    \frac{1}{\sqrt{2\pi\mu_1}}.
\end{equation}
These asymptotic estimates will be compared with exact results for one- and two-mode oscillators in Sections~\ref{sec:one}, \ref{sec:num}.

\section{Asymptotic behaviour of the invariant covariances}
\label{sec:cov}

Whereas low decoherence is important for isolating the oscillator from the   environment, this requirement conflicts with accelerating the  convergence to the invariant quantum state through enhanced dissipation, which, as mentioned at the end of Section~\ref{sec:sys},  underlies the Gaussian state generation. The system-field coupling influences not only the convergence rate (\ref{lead}), but also the Gaussian invariant state itself.   To this end, of interest is the asymptotic behaviour of the real part of the invariant covariance matrix
in (\ref{P}),
\begin{equation}
\label{Peps}
    P_\eps
    =
    \int_{\mR_+}
    \re^{tA_\eps }B_\eps B_\eps ^\rT \re^{tA_\eps^\rT} \rd t
    =
    \eps^2
    \int_{\mR_+}
    \re^{tA_\eps }\sB \sB^\rT \re^{tA_\eps^\rT} \rd t
\end{equation}
as $\eps\to 0+$, where the matrices $A_\eps$, $B_\eps$  are given by (\ref{ABeps}),  with the ALE (\ref{PALE}) taking the form
\begin{equation}
\label{PepsALE}
    A_\eps P_\eps + P_\eps A_\eps^\rT + \eps^2 \sB\sB^\rT = 0.
\end{equation}
\begin{thm}
\label{th:Plim}
Suppose the assumptions of Theorems~\ref{th:asy}, \ref{th:stab}, are satisfied along with the condition (\ref{mupos}). Then the matrix $P_\eps$ in (\ref{Peps}) has the limit
\begin{equation}
\label{Plim}
    \Pi
    :=
      \lim_{\eps \to 0+}
    P_\eps
    =
    R^{-1/2}
    \sum_{k=1}^{n/2}
    \frac{1}{\omega_k\mu_k}
    |\sB^\rT \sqrt{R} v_k|^2
    \Re
    (
    v_k  v_k^*
    )
    R^{-1/2},
\end{equation}
computed in terms of the energy matrix $R$, the eigendata from (\ref{V})--(\ref{vsym}), the matrix $\sB$ from (\ref{ABeps}) and the quantities $\mu_k$ defined in  (\ref{mu}). \hfill$\square$
\end{thm}
\begin{proof}
Under the conditions of Theorems~\ref{th:asy}, \ref{th:stab}, there exists a $\delta>0$ such that for all $\eps\in [0, \delta)$, the matrix $A_\eps$ has pairwise different eigenvalues $\lambda_1(\eps), \ldots, \lambda_n(\eps)$ in  (\ref{eigasy}) and  is diagonalisable as
\begin{equation}
\label{SAS}
  A_{\eps} = S_{\eps} \Lambda_\eps S_\eps^{-1},
\end{equation}
where $S_\eps \in \mC^{n\x n}$ is a nonsingular matrix whose columns are the corresponding eigenvectors, and
\begin{equation}
\label{Lam}
  \Lambda_\eps := \diag_{1\< k \< n} (\lambda_k(\eps)).
\end{equation}
The matrix $S_\eps$ (which is defined up to an arbitrary nonsingular diagonal right factor) can be chosen so as to inherit  from $A_\eps$ the infinite differentiability over $\eps \in [0, \delta)$ and to have the matrix $S$ from (\ref{AV}) as the limit:
\begin{equation}
\label{SSS}
    S_0
    =
    \lim_{\eps \to 0+} S_\eps = S.
\end{equation}
Since $A_\eps^\rT = A_\eps^* = S_{\eps}^{-*} \overline{\Lambda_\eps} S_\eps^*$ in view of (\ref{SAS}),   with $(\cdot)^{-*} := ((\cdot)^{-1})^*$, the ALE (\ref{PepsALE}) can be represented as
\begin{equation}
\label{QepsALE}
    \Lambda_\eps Q_\eps
    +
    Q_\eps \overline{\Lambda_\eps}
    +
    \eps^2 S_\eps^{-1} \sB\sB^\rT S_\eps^{-*} = 0,
\end{equation}
where
\begin{equation}
\label{Qeps}
    Q_\eps := (q_{jk}(\eps))_{1\< j,k\< n}:=  S_\eps^{-1} P_\eps S_{\eps}^{-*} = Q_\eps^*\succcurlyeq 0
\end{equation}
is an auxiliary matrix. The diagonal structure of the matrix $\Lambda_\eps$ in (\ref{Lam}) allows the ALE (\ref{QepsALE}) to be solved entrywise:
\begin{equation}
\label{qjk}
  q_{jk}(\eps)
  =
  -
  \frac{\eps^2}{\lambda_j(\eps) + \overline{\lambda_k(\eps)}}
  (S_\eps^{-1} \sB\sB^\rT S_\eps^{-*})_{jk},
  \qquad
  j,k=1, \ldots, n.
\end{equation}
Now, (\ref{eigasy}) implies that
\begin{equation}
\label{lim0}
    \lambda_j(\eps) + \overline{\lambda_k(\eps)}
    =
    i(\omega_j - \omega_k)
    -
    \eps^2 (\omega_j \mu_j + \omega_k \mu_k)
     + o(\eps^2),
     \qquad
     {\rm as}\
     \eps\to 0+,
\end{equation}
and hence,
\begin{align}
\nonumber
    \lambda_j(0) + \overline{\lambda_k(0)}
    & =
    \lim_{\eps \to 0+}
    (\lambda_j(\eps) + \overline{\lambda_k(\eps)})\\
\label{lim1}
     & = i(\omega_j-\omega_k)
    \ne 0,
    \qquad
    1\< j\ne k \< n,
\end{align}
in view of (\ref{omdiff}),
while
\begin{equation}
\label{lim2}
    \lambda_k(\eps) + \overline{\lambda_k(\eps)}
    =
    2\Re \lambda_k(\eps)
    \sim
    -
    2 \omega_k \mu_k \eps^2,
    \qquad
    {\rm as}\
    \eps\to 0+.
\end{equation}
By combining (\ref{qjk}) with (\ref{lim1}), (\ref{lim2}), (\ref{SSS}), it follows that
\begin{equation}
\label{qlim}
  \lim_{\eps \to 0+}
  q_{jk}(\eps)
  =
  \frac{\delta_{jk}}{2 \omega_k\mu_k}
  (S^{-1} \sB\sB^\rT S^{-*})_{jk},
  \qquad
  j,k=1, \ldots, n,
\end{equation}
where $\delta_{jk}$ is the Kronecker delta. Since the matrix $P_{\eps}$ can be  recovered from (\ref{Qeps}) as
\begin{equation}
\label{SQS}
    P_\eps = S_\eps Q_\eps S_\eps^*,
\end{equation}
then, in view of (\ref{SSS}), (\ref{qlim}), it has the following limit
\begin{align*}
    \lim_{\eps \to 0+}
    P_\eps
    & =
    \frac{1}{2}
    S
    \diag_{1\< k \< n}
    \Big(\frac{1}{\omega_k\mu_k}
    (S^{-1} \sB\sB^\rT S^{-*})_{kk}
  \Big)
     S^*\\
    & =
    \frac{1}{2}
    R^{-1/2}
    V
    \diag_{1\< k \< n}
    \Big(\frac{1}{\omega_k\mu_k}
    (V^* \sqrt{R} \sB\sB^\rT \sqrt{R}V)_{kk}
  \Big)
     V^* R^{-1/2}\\
    & =
    \frac{1}{2}
    R^{-1/2}
    V
    \diag_{1\< k \< n}
    \Big(\frac{1}{\omega_k\mu_k}
    |\sB^\rT \sqrt{R}v_k|^2
  \Big)
     V^* R^{-1/2}\\
    & =
    \frac{1}{2}
    R^{-1/2}
    \sum_{k=1}^n
    \frac{1}{\omega_k\mu_k}
    |\sB^\rT \sqrt{R} v_k|^2
    v_k  v_k^*
    R^{-1/2},
\end{align*}
which uses the structure of the matrices $S$, $V$  from (\ref{AV}), (\ref{V}) and yields (\ref{Plim}) due to the symmetries (\ref{vsym}), (\ref{ommu3}). 
\end{proof}

The limit (\ref{Plim})  depends on the coupling shape matrix $\sM$ only through the matrix $\sB$ in (\ref{ABeps}) and the quantities  (\ref{mu}) and is invariant under the scaling transformation $\sM \mapsto \sigma \sM$ for any $\sigma\in \mR\setminus \{0\}$. Also,
$$
    \Pi+i\Theta \succcurlyeq 0,
$$
which is inherited from the positive semi-definiteness  of the quantum covariance matrices (\ref{PT}).

\section{Asymptotic estimates and exact results for a one-mode oscillator}
\label{sec:one}

In the one-mode case, the results of Section~\ref{sec:asy} hold not only in the asymptotic sense, but also beyond the weak-coupling assumption. More precisely, if $n=2$, the real antisymmetric $(2\x 2)$-matrices $\Theta$ and $\sM^\rT J \sM$ can be represented as
\begin{equation}
\label{thetagamma}
    \Theta = \theta \bJ,
    \qquad
    \sM^\rT J \sM = \gamma \bJ
\end{equation}
in terms of $\theta, \gamma  \in \mR$ and the matrix $\bJ$ from (\ref{bJ}). In particular, if the system variables of the OQHO are the conjugate position and momentum operators $q$ and $p:= -i\d_q$ on the Schwartz space, then
\begin{equation}
\label{1/2}
    \theta = \frac{1}{2}
\end{equation}
(see, for example,    \cite{S_1994}). Since the matrix $\Theta$ in (\ref{thetagamma}) satisfies $\det \Theta = \theta^2$ and hence,  $\theta \ne 0$  in view of the condition $\det \Theta \ne 0$ in (\ref{typ}), then the relation (\ref{1/2}) can be achieved by rescaling one of the system variables (for example, $X_1$) as $X_1\mapsto \frac{1}{2\theta} X_1$ while leaving the other unchanged ($X_2\mapsto X_2$). This allows (\ref{thetagamma}) to be assumed to  hold along with (\ref{1/2}) in what follows, without loss of generality. Therefore,
since $\bJ^2 = -I_2$,    the matrix $A_\eps$ in (\ref{ABeps}) takes the form
\begin{equation}
\label{Athetanu}
    A_\eps
    =
    A_0 + 2\eps^2 \theta \gamma \bJ^2
    = A_0 -\eps^2 \gamma I_2,
\end{equation}
where
\begin{align}
\nonumber
    A_0
    & = 2\theta \bJ R\\
\nonumber
     & = R^{-1/2}\sqrt{R} \bJ \sqrt{R} \sqrt{R} \\
\label{A00}
     & =
    \sqrt{\det R}\,  R^{-1/2}\bJ \sqrt{R}
\end{align}
in view of the condition $R\succ 0$ in (\ref{typ}) and the identity $N \bJ N^\rT = (\det N) \bJ$ for any $N\in \mC^{2\x 2}$. With the matrix $R^{-1/2}\bJ \sqrt{R}$ being isospectral to $\bJ$,
it follows from  (\ref{bJ}), (\ref{A00}) that the eigenfrequencies  of $A_0$ are specified by 
\begin{equation}
\label{om12}
    \omega_{1,2} = \pm \omega,
    \qquad
    \omega := \sqrt{\det R} >0.
\end{equation}
Therefore, the spectrum of the matrix $A_\eps$ in (\ref{Athetanu}) is given by
\begin{equation}
\label{lam12}
    \lambda_{1,2}(\eps)
    =
    -\eps^2 \gamma \pm i\omega
    =
    \omega (-\mu \eps^2 \pm i),
\end{equation}
and its quadratic dependence on $\eps$ is valid not only asymptotically as in (\ref{eigasy}) of  Theorem~\ref{th:asy},  with
\begin{equation}
\label{mu12}
    \mu_{1,2} = \pm \mu,
    \qquad
    \mu:= \frac{\gamma}{\omega},
\end{equation}
but for any coupling strength $\eps$. In view of (\ref{lam12}), (\ref{mu12}), the condition $\mu>0$, which secures the Hurwitz property of $A_\eps$ for all $\eps > 0$, is equivalent to
 \begin{equation}
 \label{gammapos}
    \gamma > 0.
\end{equation}
The relations (\ref{Athetanu}), (\ref{A00}) allow the corresponding matrix exponential in (\ref{taud}) to be represented as
\begin{equation}
\label{exp}
    \re^{\tau A_\eps}
    =
    \re^{-\eps^2 \gamma \tau}
    \re^{\tau A_0},
\end{equation}
where
\begin{equation}
\label{RR}
    \re^{\tau A_0}
    =
    R^{-1/2}
    \re^{\tau \sqrt{\det R}\,  \bJ }
    \sqrt{R}
    =
    R^{-1/2}
    \Sigma(\omega \tau)
    \sqrt{R}
\end{equation}
is a $T$-periodic function of time $\tau$, with the period
\begin{equation}
\label{Tom}
    T = \frac{2\pi}{\omega}
\end{equation}
expressed in terms of the eigenfrequency $\omega$ from (\ref{om12})
in accordance with (\ref{T*}).
Here, we have used the property that $\bJ$ in (\ref{bJ}) is an infinitesimal  generator of the group of planar rotation matrices:
\begin{equation}
\label{rot}
    \re^{\phi \bJ}
    =
    \begin{bmatrix}
      \cos \phi & \sin \phi\\
      -\sin \phi & \cos \phi
    \end{bmatrix}
    =:
    \Sigma(\phi),
    \qquad
    \phi \in \mR.
\end{equation}
In (\ref{RR}),
use is also made of the relation $f(SNS^{-1}) = Sf(N)S^{-1}$ for equally dimensioned complex matrices $N$, $S$,  with $\det S\ne 0$ and any  function $f$, holomorphic in a neighbourhood of the spectrum of $N$; see, for example, \cite{H_2008}.  The $T$-periodicity of $\re^{\tau A_0}$ with respect to $\tau$ in (\ref{RR}) implies that $\re^{T A_0} = I_2$, and hence,
\begin{equation}
\label{exp1}
    \|\re^{T A_\eps}\Theta \|_\rF
    =
    \re^{-\eps^2 \gamma T}
    \|\Theta\|_\rF
    =
    \re^{-2\pi \mu \eps^2}
    \|\Theta\|_\rF
\end{equation}
in view of (\ref{mu12}), (\ref{exp}), (\ref{Tom}), where $\|\Theta\|_\rF = \frac{1}{2}\|\bJ\|_\rF = \frac{1}{\sqrt{2}}$ (although the latter is irrelevant for quantifying the decay).   Therefore, the condition
\begin{equation}
\label{eps<}
    \eps < \frac{1}{\sqrt{2\pi \mu}}
\end{equation}
is necessary for the decoherence time $\tau_*$ in  (\ref{taud}) to satisfy $\tau_* > T$.
Indeed, if $\eps \>  \frac{1}{\sqrt{2\pi \mu}}$, then  (\ref{exp1}) yields
$
    \|\re^{T A_\eps}\Theta \|_\rF
    \<
    \frac{1}{\re}
    \|\Theta\|_\rF
$ and hence, $\tau_* \< T$ in view of (\ref{taud}). This implication establishes the necessity of (\ref{eps<}) for $\tau_* > T$. Furthermore, the condition (\ref{eps<}) becomes necessary and sufficient if the standard Frobenius norm $\|\cdot\|_\rF$  in  (\ref{taud}) is replaced with its weighted version
$$
    \|N\|_R := \|\sqrt{R}N\|_\rF = \sqrt{\Tr (N^\rT R N)},
$$
associated with the energy matrix $R\succ 0$ of the OQHO. Indeed, in application of this norm to the CCR matrices, it follows from (\ref{exp}), (\ref{RR}) and the orthogonality of the rotation matrix (\ref{rot}) that
\begin{align}
\nonumber
    \|\re^{\tau A_\eps} \Theta \|_R
    & =
    \re^{-\eps^2 \gamma \tau}
    \|\re^{\tau A_0} \Theta \|_R\\
\nonumber
    & =
    \re^{-\eps^2 \gamma\tau}
    \|    R^{-1/2}
    \Sigma(\omega \tau)
    \sqrt{R}\Theta \|_R    \\
\nonumber
    & =
    \re^{-\eps^2 \gamma \tau}
    \|
    \Sigma(\omega \tau)
    \sqrt{R}\Theta \|_\rF    \\
\label{RF}
    & =
    \re^{-\eps^2 \gamma \tau}
    \|
    \sqrt{R}\Theta \|_\rF    =
    \re^{-\eps^2 \gamma \tau}
    \|
    \Theta \|_R.
\end{align}
Due to the factor $\re^{-\eps^2 \gamma \tau}$ on the right-hand side of  (\ref{RF}), which decays under the condition (\ref{gammapos}),
an appropriately   modified decoherence time $\tau_R$  (with the weighted norm $\|\cdot \|_R$ instead of $\|\cdot\|_\rF$ in (\ref{taud})) takes the form
$$
    \tau_R = \frac{\eps^{-2} }{\gamma}
$$
and coincides with (\ref{tauhat}) since $\omega \mu = \gamma$ in view of (\ref{mu12}). Accordingly, the low decoherence requirement $\tau_R \gg T$ for the one-mode oscillator  on the time scale of (\ref{Tom}) becomes (\ref{eps<<}) as a strong version of (\ref{eps<}).

For completeness, we will also discuss the behaviour of the invariant real covariance matrix $P_\eps$ from  (\ref{Peps}). In view of (\ref{Athetanu})  in the one-mode case being considered,  the diagonalising matrix $S_\eps$ in (\ref{SAS}) does not depend on $\eps$ and coincides with the matrix $S$ from (\ref{AV}), where
\begin{equation}
\label{V1}
    V
    =
    \begin{bmatrix}
      v_1 & v_2
    \end{bmatrix},
    \qquad
    v_{1,2} :=
    \frac{1}{\sqrt{2}}
    \begin{bmatrix}
      1 \\
      \pm i
    \end{bmatrix}
\end{equation}
is formed from the orthonomal eigenvectors of the matrix $\bJ$ from (\ref{bJ}),  in accordance with (\ref{V})--(\ref{vsym}) and the relations
$$
    2\sqrt{R}\Theta \sqrt{R} = \sqrt{R}\bJ \sqrt{R} = \omega \bJ ,
$$
where $\omega$ is the eigenfrequency  given by (\ref{om12}).
Hence, the representation (\ref{SQS}) takes the form
\begin{equation}
\label{SQS1}
    P_\eps
    =
    S Q_\eps S^*
    =
    R^{-1/2}
    \sum_{j,k=1}^2
    q_{jk}(\eps)
    v_j v_k^*
    R^{-1/2},
\end{equation}
where the matrix $Q_\eps := (q_{jk}(\eps))_{1\< j,k\< 2}$ in (\ref{Qeps}) is computed by combining (\ref{qjk}) and (\ref{lam12}) as
\begin{equation}
\label{qjk1}
  q_{jk}(\eps)
  =
  -
  \frac{\eps^2}{\lambda_j(\eps) + \overline{\lambda_k(\eps)}}
  v_j^* \sqrt{R}\sB\sB^\rT \sqrt{R} v_k,
  \qquad
  j,k=1, 2,
\end{equation}
with
\begin{align}
\label{lam123_1}
    \lambda_k(\eps)+\overline{\lambda_k(\eps)}
    & = 2\Re \lambda_k(\eps) = -2\eps^2 \gamma,\\
\label{lam123_2}
    \lambda_k(\eps)+\overline{\lambda_{3-k}(\eps)}
    & =
    2(-\eps^2 \gamma \pm i\omega),
    \qquad
    k = 1,2,
\end{align}
so that the asymptotic relations (\ref{lim0}), (\ref{lim2}) are exact in the one-mode case.
Substitution of (\ref{qjk1})--(\ref{lam123_2}) into (\ref{SQS1}) leads to
\begin{align}
\nonumber
    P_\eps
     = &
    \underbrace{\frac{1}{2\gamma}
    |\sB^\rT \sqrt{R}v_1|^2 R^{-1}}_{\Pi}\\
\label{Peps1}
    & +
    \underbrace{\eps^2
    R^{-1/2}
        \Re
        \Big(
            \frac{1}{\eps^2 \gamma - i\omega}
            (v_1^*\sqrt{R}\sB \sB^\rT \sqrt{R} v_2)
            v_1v_2^*
        \Big)
    R^{-1/2}}_{O(\eps^2)}.
\end{align}
Here, the first term uses the property $\Re(v_kv_k^*) = \frac{1}{2} I_2$ of the eigenvectors in  (\ref{V1}) and describes the weak-coupling limit (\ref{Plim}) for the invariant covariance matrix, while the second term in (\ref{Peps1}) calculates the remainder in closed form.

\section{A numerical example with a two-mode oscillator}
\label{sec:num}

For a numerical illustration, consider a two-mode ($n=4$) OQHO with a three-channel ($m=6$) external field. The system variables are two  conjugate position-momentum pairs $(q_1,p_1)$, $(q_2,p_2)$ which are  assembled into the vector
$X:= [
  q_1,
  q_2,
  p_1,
  p_2
]^\rT$ in (\ref{XCCR}) with the CCR matrix
\begin{equation}
\label{TJI}
    \Theta = \frac{1}{2} \bJ\ox I_2,
\end{equation}
where $\bJ$ is given by (\ref{bJ}). The energy and coupling shape matrices are generated randomly (subject to the conditions of Theorem~\ref{th:stab}):
\begin{equation}
\label{RM1}
    R :=
    {\scriptsize\begin{bmatrix}
    1.3522 &   1.2976 &  -0.5225 &  -2.7819\\
    1.2976 &   6.0400 &   1.1741 &  -1.4901\\
   -0.5225 &   1.1741 &   2.2582 &   1.5566\\
   -2.7819 &  -1.4901 &   1.5566 &   6.7739
    \end{bmatrix}},
    \quad
    \sM:=
    {\scriptsize\begin{bmatrix}
   -1.7423  & -1.6364 &  -0.3020 &  -0.3483\\
    0.2053  &  0.0173 &   1.8136 &   1.4126\\
    1.1929  &  0.8284 &   0.9149 &   1.5024\\
   -0.8028  &  0.2177 &  -0.0571 &   0.7304\\
   -1.2656  & -1.9092 &   1.3094 &   0.4908\\
   -0.1493  & -0.5368 &  -1.0447 &  -0.5861
   \end{bmatrix}}.
\end{equation}
The positive eigenfrequencies (\ref{sym}) of the matrix $A_0$ for the uncoupled oscillator in (\ref{A0}) are
$$
   \omega_1 = 7.2046, \qquad    \omega_2 = 0.3729,
$$
and the corresponding orthonormal eigenvectors from (\ref{veig}) are
$$
    v_1
    =
    {\scriptsize\begin{bmatrix}
    0.1046 - 0.2450i\\
   0.6592 + 0.0000i\\
   0.2323 + 0.1748i\\
  -0.0234 + 0.6399i
    \end{bmatrix}},
    \qquad
    v_2
    =
    {\scriptsize\begin{bmatrix}
    0.6550 + 0.0000i\\
   -0.1053 - 0.2332i\\
   0.0283 + 0.6440i\\
  0.2431 - 0.1759i
    \end{bmatrix}}.
$$
The calculation of the quantities (\ref{mu}), associated with these eigenvectors, yields
$$
   \mu_1 =  0.1765,
   \qquad
   \mu_2 =
    5.2214,
$$
and the leading coefficient of the asymptotic approximation on the right-hand side of (\ref{lead}) takes the value
$$
    \min_{k=1,2}(\omega_k \mu_k)=1.2714.
$$
A comparison of the exact value of the leading Lyapunov exponent with its approximation is provided by Fig.~\ref{fig:expAeps}. 
\begin{figure}[htbp]
\begin{center}
\includegraphics[width=10cm]{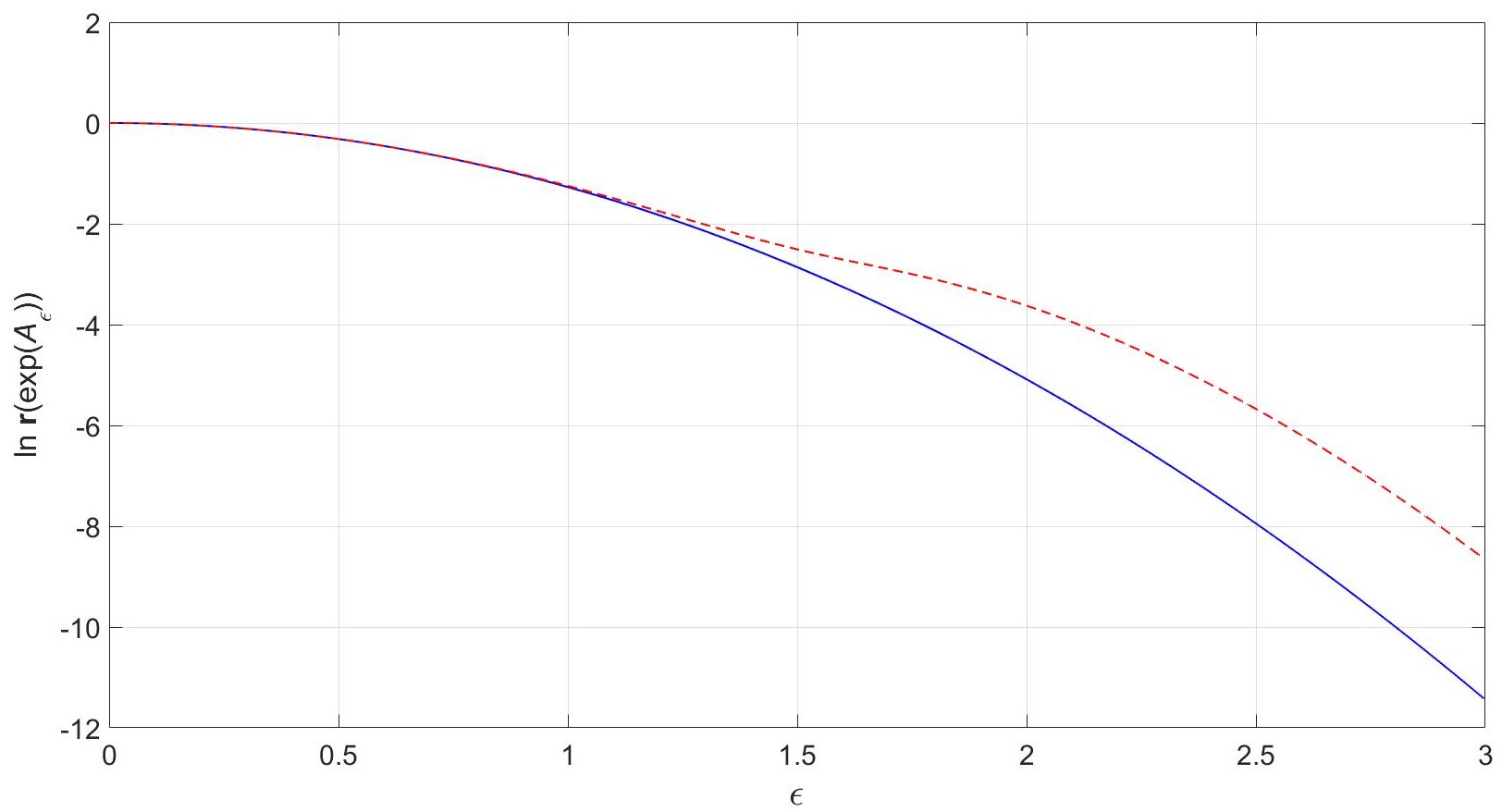}
\caption{The exact value $\ln \br(\re^{A_\eps})$ of the leading Lyapunov exponent for the matrix $A_\eps$ of the OQHO in (\ref{ABeps}) as a function of the coupling strength parameter $\eps$ (dashed line)  in comparison with its asymptotic approximation $- \eps^2 \min_{k = 1, 2}(\omega_k\mu_k)$ from (\ref{lead})  (solid line) for the two-mode OQHO example (\ref{TJI}), (\ref{RM1}). }
\label{fig:expAeps}
\end{center}
\end{figure} 
The resulting thresholds for the   coupling strength parameter in (\ref{epsmax}), (\ref{mumax}) are
$$
    \wh{\eps} = 0.2161,
    \qquad
    \wt{\eps} = 0.1746 .
$$
As Fig.~\ref{fig:expAeps} shows, the requirement $\eps\ll \wh{\eps}$ (or $\eps\ll \wt{\eps}$) for securing a low decoherence level for the OQHO, obtained in the weak-coupling limit, is consistent with the validity range of this approximation in this example.

\section{Decoherence control by interconnection}
\label{sec:two}

We will now apply the results of the previous sections to a decoherence control setting for two OQHOs (interpreted, for example,  as a quantum plant and a quantum controller) which, in addition to their interaction with external bosonic fields,   are coupled to each other in a coherent (measurement-free) fashion. The latter involves a direct energy coupling and an indirect field-mediated coupling \cite{ZJ_2011a} as shown in
Fig.~\ref{fig:system}. 
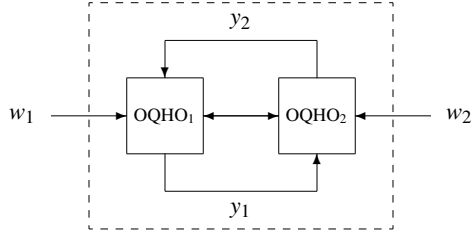
\begin{figure}[htbp]
\centering
\unitlength=1mm
\linethickness{0.2pt}
\begin{picture}(50.00,30.00)
    \put(5,0){\dashbox(40,30)[cc]{}}
    \put(10,10){\framebox(10,10)[cc]{\scriptsize OQHO${}_1$}}
    \put(30,10){\framebox(10,10)[cc]{\scriptsize OQHO${}_2$}}
    \put(0,15){\vector(1,0){10}}
    \put(50,15){\vector(-1,0){10}}
     \put(35,20){\line(0,1){5}}
    \put(35,25){\line(-1,0){20}}
    \put(15,25){\vector(0,-1){5}}

     \put(15,10){\line(0,-1){5}}
    \put(15,5){\line(1,0){20}}
    \put(35,5){\vector(0,1){5}}

    \put(20,15){\vector(1,0){10}}
    \put(30,15){\vector(-1,0){10}}
    \put(-2,15){\makebox(0,0)[rc]{\small $w_1$}}
    \put(25,27.5){\makebox(0,0)[cc]{\small $y_2$}}
    \put(52,15){\makebox(0,0)[lc]{\small $w_2$}}
    \put(25,2.5){\makebox(0,0)[cc]{\small $y_1$}}
\end{picture}
\caption{
    An interconnection of two OQHOs, which have external input quantum Wiener processes  $w_1$, $w_2$  and interact with each other both directly (through the energy coupling shown as a double arrow) and indirectly in a field-mediated fashion through quantum Ito processes $y_1$, $y_2$ at their corresponding outputs.
}
\label{fig:system}
\end{figure}
In this fully quantum feedback interconnection,
the external fields  are modelled by
$\frac{m_1}{2}$- and $\frac{m_2}{2}$-channel quantum Wiener processes  $w_1$, $w_2$ (with even $m_1$, $m_2$) on symmetric Fock spaces $\fF_1$, $\fF_2$, respectively. These processes form an augmented quantum Wiener process
\begin{equation}
\label{womega}
    W
    :=
    \begin{bmatrix}
        w_1\\
        w_2
    \end{bmatrix}
\end{equation}
on the composite Fock space $\fF:= \fF_1\ox \fF_2$ and, similarly to (\ref{Omega}), have the following Ito tables
\begin{equation}
\label{wwww}
    \rd w_k\rd w_k^{\rT} = \Omega_k \rd t,
    \qquad
    \rd W \rd W^{\rT} = \Omega \rd t   ,
    \qquad
    k = 1, 2,
\end{equation}
with the quantum Ito matrices $\Omega_1$, $\Omega_2$, $\Omega$ given by
\begin{align}
\nonumber
    \Omega_k
    & := I_{m_k} + iJ_k,
    \qquad
    J_k:= \bJ \ox I_{m_k/2},\\
\label{Om12}
    \Omega
    & :=
    \blockdiag_{k=1,2}(\Omega_k)
    =
    I_m + iJ,
    \qquad
    J:=
    \blockdiag_{k=1,2}(J_k)    ,
    \qquad
    m:= m_1+m_2,
\end{align}
where $\bJ$ is the matrix  from (\ref{bJ}). The constituent OQHOs are endowed with initial spaces $\fH_1$, $\fH_2$ and even numbers $n_1$, $n_2$ of dynamic variables, acting  on the space $\fH:= \fH_0\ox \fF$ (with $\fH_0:= \fH_1 \ox \fH_2$ the initial space of the composite system) and  assembled into vectors $x_1$, $x_2$ which form the augmented vector
\begin{equation}
\label{xx}
    X
    :=
    \begin{bmatrix}
        x_1\\
        x_2
    \end{bmatrix}.
\end{equation}
In accordance with (\ref{XCCR}), the one-point CCR matrices $\Theta_k = -\Theta_k^\rT \in \mR^{n_k\x n_k}$ of the OQHOs specify the relations
\begin{equation}
\label{Theta12}
    [x_k,x_k^{\rT}] = 2i\Theta_k,
    \qquad
    [X,X^{\rT}] = 2i\Theta,
    \qquad
    \Theta
    :=
    \blockdiag_{k=1,2}(\Theta_k),
    \qquad
    k = 1, 2.
\end{equation}
Here, the block-diagonal structure of the augmented CCR matrix $\Theta$ comes from the commutativity
\begin{equation}
\label{x12comm}
    [x_1, x_2^\rT] = 0
\end{equation}
of the dynamic variables of the constituent OQHOs  as operators  acting initially (at time $t=0$)  on different spaces $\fH_1$, $\fH_2$ and since the system-field evolution preserves the CCRs. The direct coupling of the OQHOs is modelled by complementing the individual Hamiltonians $\frac{1}{2} x_k^\rT R_k x_k$ of the OQHOs, specified by their energy matrices $R_k = R_k^\rT \in \mR^{n_k \x n_k}$ for $k=1,2$,  with an additional term
\begin{equation}
\label{Hint}
    x_1^\rT R_{12} x_2 = x_2^\rT R_{21} x_1
\end{equation}
parameterised by $R_{12} = R_{21}^\rT \in \mR^{n_1\x n_2}$. The equality (\ref{Hint}) uses the commutativity (\ref{x12comm}).
The indirect coupling of the OQHOs in Fig.~\ref{fig:system} is mediated by their $\frac{p_1}{2}$- and $\frac{p_2}{2}$-channel output fields $y_1$, $y_2$ (with $p_1$, $p_2$ even) which are quantum Ito processes whose
Heisenberg dynamics are governed by the linear QSDEs
\begin{align}
\label{x}
    \rd x_k
    & =
    (A_k x_k + F_k x_{3-k}) \rd t  +  B_k \rd w_k  + E_k \rd y_{3-k} ,\\
\label{y}
    \rd y_k
    & =
    C_k x_k \rd t  +  D_k \rd w_k ,
    \qquad
    k = 1, 2.
\end{align}
The matrices
\begin{equation*}
\label{ABCDEF}
    A_k\in \mR^{n_k\x n_k},
    \
    B_k\in \mR^{n_k\x m_k},
    \
    C_k\in \mR^{p_k\x n_k},
    \
    D_k\in \mR^{p_k\x m_k},
    \
    E_k\in \mR^{n_k\x p_{3-k}},\
    F_k\in \mR^{n_k\x n_{3-k}}
\end{equation*}
are parameterised as
\begin{align}
\label{Ak}
    A_k
     & =
    2\Theta_k(R_k + M_k^{\rT}J_k M_k + L_k^{\rT}\wt{J}_{3-k}L_k),\\
\label{B_k}
    B_k
    & = 2\Theta_k M_k^{\rT},\\
\label{Ck}
    C_k & =2D_kJ_k M_k,\\
\label{Ek}
    E_k & = 2\Theta_k L_k^{\rT},\\
\label{Fk}
    F_k & = 2\Theta_k R_{k,3-k},
\end{align}
where
\begin{equation}
\label{tJk}
    \wt{J}_k:= D_kJ_kD_k^{\rT},
    \qquad
    k = 1,2,
\end{equation}
which can be obtained by using the quantum feedback network formalism \cite{GJ_2009,JG_2010}.
Here, 
$M_k\in \mR^{m_k \x n_k}$,  $L_k\in \mR^{p_{3-k} \x n_k}$ are the matrices of coupling of the $k$th  OQHO to the corresponding external input field $w_k$ and the output $y_{3-k}$ of the other OQHO, respectively.  Also,
each of the feedthrough matrices $D_k$ in (\ref{y}) is associated with a particular selection of the component output fields  and is formed from conjugate pairs of rows of a permutation matrix of order $m_k$, so that $p_k\< m_k$, with
\begin{equation*}
\label{DDI}
    D_kD_k^{\rT} = I_{p_k},
    \qquad
    k = 1,2.
\end{equation*}
The combined QSDEs
 (\ref{x}), (\ref{y}) describe an augmented OQHO with $n:= n_1+n_2$ dynamic variables in (\ref{xx}),  which have the one-point CCR matrix $\Theta$ in (\ref{Theta12}) and are driven by the $\frac{m}{2}$-channel quantum Wiener process $W$ in (\ref{womega})--(\ref{Om12})  according to the QSDE (\ref{dX}) with the matrices
 \begin{equation}
\label{cAB}
    A
    =
    \begin{bmatrix}
        A_1 & F_1+ E_1C_2\\
        F_2 +E_2C_1 & A_2
    \end{bmatrix},
    \qquad
    B
    =
    \begin{bmatrix}
        B_1 & E_1D_2\\
        E_2D_1 & B_2
    \end{bmatrix}.
\end{equation}
The energy and coupling matrices of the resulting closed-loop OQHO in Fig.~\ref{fig:system} are computed by substituting (\ref{Ak})--(\ref{tJk}) into (\ref{cAB}) and comparing the result with (\ref{AB}) in combination with (\ref{Om12}), (\ref{Theta12}):
\begin{align}
\label{Rclos}
    R & =
    \begin{bmatrix}
      R_1                                       & R_{12}+\frac{1}{2}(L_1^{\rT}C_2 +C_1^{\rT}L_2)\\
      R_{21}+\frac{1}{2}(C_2^{\rT}L_1+L_2^{\rT}C_1)   & R_2
    \end{bmatrix}
    =
    R_0
    +
    \wt{R}  ,\\
\label{Mclos}
    M  & =
    \begin{bmatrix}
      M_1 & D_1^{\rT}L_2 \\
      D_2^{\rT}L_1 & M_2
    \end{bmatrix}.
\end{align}
Here,
\begin{equation}
\label{R0}
  R_0 :=
  \begin{bmatrix}
    R_1 & R_{12}\\
    R_{21} & R_2
  \end{bmatrix}
\end{equation}
is the energy matrix which the closed-loop system would have if the indirect field-mediated coupling between the constituent OQHOs were removed (that is, in the case when $L_1 = 0$, $L_2 = 0$). Accordingly, the additional term
\begin{equation}
\label{Rt}
    \wt{R}
    :=
    R-R_0
    =
    \begin{bmatrix}
      0                                       & L_1^{\rT}D_2J_2 M_2 -M_1^\rT J_1 D_1^\rT L_2\\
      L_2^{\rT}D_1J_1 M_1 -M_2^\rT J_2 D_2^\rT L_1  & 0
    \end{bmatrix}
\end{equation}
in  (\ref{Rclos}) originates from the field-mediated coupling between the OQHOs. With the feedthrough matrices $D_1$, $D_2$ being fixed, the closed-loop coupling matrix $M$ in (\ref{Mclos}) is a linear function of the coupling matrices $M_1$, $L_1$, $M_2$, $L_2$, while $\wt{R}$ in (\ref{Rt}) depends  on them in a quadratic fashion. The role of the matrix (\ref{A0}) for the composite system will now be played by
\begin{equation}
\label{A0clos}
  A_0
  :=
  2\Theta R_0
  =
  2
  \begin{bmatrix}
    \Theta_1 R_1 & \Theta_1 R_{12}\\
    \Theta_2 R_{21} & \Theta_2 R_2
  \end{bmatrix}.
\end{equation}
In application to the block-diagonal matrix $\Theta$ in (\ref{Theta12}) and the matrix $R_0$  in (\ref{R0}), the condition  (\ref{typ}) is equivalent to
\begin{equation}
\label{typclos}
    R_0 \succ 0 ,
    \qquad
    \det \Theta_k \ne 0,
    \qquad
    k = 1,2,
\end{equation}
and guarantees that the spectrum of the matrix $A_0$ in (\ref{A0clos}) is purely imaginary and specified by the eigenfrequencies $\omega_1, \ldots, \omega_n \in \mR$ as before. The relations  (\ref{AV}) are modified as
\begin{equation}
\label{AVclos}
    A_0 = i S \mho S^{-1},
    \qquad
    S := R_0^{-1/2}V,
    \qquad
    \mho := \diag_{1\< k \< n} (\omega_k),
\end{equation}
where, this time,  the unitary matrix $V$ in (\ref{V}) consists of the eigenvectors $v_1, \ldots, v_n$ of the Hermitian matrix
\begin{equation}
\label{VVclos}
    -2i\sqrt{R_0}\Theta \sqrt{R_0} = V\mho V^*.
\end{equation}
In accordance with the weak-coupling framework of Section~\ref{sec:asy}, suppose the coupling matrices
\begin{equation}
\label{Mepsclos}
    M_k := \eps \sM_k,
    \qquad
    L_k := \eps \sL_k,
    \qquad
    k = 1,2,
\end{equation}
are specified by a coupling strength parameter $\eps \> 0$ and appropriately dimensioned coupling shape matrices $\sM_k$, $\sL_k$. Then (\ref{Mclos}), (\ref{Rt}) are represented as  $M = \eps \sM$, $\wt{R} = \eps^2 \sR$ in terms of
\begin{align}
\label{sMclos}
    \sM
    & :=
    \begin{bmatrix}
      \sM_1 & D_1^{\rT}\sL_2 \\
      D_2^{\rT}\sL_1 & \sM_2
    \end{bmatrix},\\
\label{sRt}
    \sR
    &
    :=
    \begin{bmatrix}
      0                                       & \sL_1^{\rT}D_2J_2 \sM_2 -\sM_1^\rT J_1 D_1^\rT \sL_2\\
      \sL_2^{\rT}D_1J_1 \sM_1 -\sM_2^\rT J_2 D_2^\rT \sL_1  & 0
    \end{bmatrix},
\end{align}
with $\sR = \sR^\rT \in \mR^{n\x n}$,
and hence, the matrices (\ref{AB}) of the closed-loop OQHO depend on $\eps$ as
\begin{equation}
\label{ABepsclos}
    A_\eps
    =
    A_0 + 2\eps^2 \Theta (\sR + \sM^\rT J \sM),
    \qquad
     B_\eps = \eps \sB,
     \qquad
     \sB := 2\Theta \sM^\rT,
\end{equation}
with $A_0$ given by (\ref{A0clos}). 
The following theorem provides an appropriate adaptation of Theorem~\ref{th:asy}.

\begin{thm}
\label{th:asyclos}
Suppose the CCR and energy matrices of the constituent OQHOs satisfy (\ref{typclos}). Also, suppose the eigenfrequencies $\omega_1,\ldots, \omega_n$  of the matrix $A_0$ in  (\ref{A0clos})  are   pairwise different.
Then for all sufficiently small values of the coupling strength parameter $\eps$ in (\ref{Mepsclos}), the matrix $A_\eps$ in (\ref{ABepsclos}) has different eigenvalues $\lambda_1(\eps), \ldots, \lambda_n(\eps)$ which, appropriately numbered, behave asymptotically as
\begin{equation}
\label{eigasyclos}
  \lambda_k(\eps)
  =
  \omega_k(i(1+\eps^2 \sigma_k) - \eps^2\mu_k) + o(\eps^2),
  \qquad
  {\rm as}\
  \eps \to 0+.
\end{equation}
Here,
\begin{align}
\label{muclos}
  \mu_k
  & := -i v_k^* R_0^{-1/2} \sM^\rT J \sM R_0^{-1/2} v_k,\\
\label{sigclos}
  \sigma_k
  & :=  v_k^* R_0^{-1/2} \sR R_0^{-1/2} v_k
  \qquad
  k = 1, \ldots, n,
\end{align}
are real-valued quantities using
the columns $v_1, \ldots, v_n$ of the unitary matrix (\ref{V}), associated with (\ref{VVclos}),  and the matrices $R_0$, $\sM$, $\sR$  from (\ref{R0}), (\ref{sMclos}), (\ref{sRt}). \hfill$\square$
\end{thm}
\begin{proof}
The relation (\ref{eigasyclos}) is established by repeating the proof of Theorem~\ref{th:asy} almost verbatim,  except that, in application to the matrix $A_\eps$ in (\ref{ABepsclos}), the similarity transformation (\ref{SS}) is modified as 
\begin{align}
\nonumber
    S^{-1}A_\eps S
    & =
    i \mho  + 2\eps^2 S^{-1}\Theta (\sR + \sM^\rT J \sM) S \\
\nonumber
    & =
    i \mho  + 2\eps^2 V^*\sqrt{R_0}\Theta (\sR + \sM^\rT J \sM)R_0^{-1/2}V\\
\nonumber
    & =
    i \mho  + 2\eps^2 V^*\sqrt{R_0}\Theta\sqrt{R_0}R_0^{-1/2} (\sR + \sM^\rT J \sM)R_0^{-1/2}V\\
\nonumber
    & =
    i \mho  + i\eps^2 V^*V\mho V^*R_0^{-1/2} (\sR + \sM^\rT J \sM)R_0^{-1/2}V\\
\label{SSclos}
    & =
    i \mho  + i\eps^2 \mho V^*R_0^{-1/2} (\sR + \sM^\rT J \sM)R_0^{-1/2}V
\end{align}
by using (\ref{AVclos}), (\ref{VVclos}). Accordingly, (\ref{kk}), leading to (\ref{muclos}) with $R_0$ from (\ref{R0}) instead of $R$, is complemented by 
$$
    (i\mho V^*R_0^{-1/2} \sR R_0^{-1/2}V)_{kk}
    =
    i\omega_k v_k^*R_0^{-1/2} \sR R_0^{-1/2}v_k
    =
    i\omega_k \sigma_k
$$
in view of (\ref{sigclos}), which,  together with (\ref{SSclos}),  yields (\ref{eigasyclos}).   The quantities $\sigma_k$ are also real-valued since $\sR$ in (\ref{sRt}) is a real symmetric matrix. 
\end{proof}

The asymptotic relations (\ref{eigasyclos}) show that the quantities $\mu_1, \ldots, \mu_n$ are responsible for the stability of the closed-loop OQHO for all sufficiently small values of the coupling strength $\eps>0$ and participate in the Lyapunov exponents and decoherence time estimates of Section~\ref{sec:asy}. More precisely, this stability is secured by
the condition (\ref{ommu1}) from the proof of Theorem~\ref{th:stab}. The representation (\ref{muclos}) allows $\mu_1, \ldots, \mu_n$ to be found as the diagonal entries of the Hermitian matrix
\begin{equation}
\label{K}
    K:= -i V^*  R_0^{-1/2}  \sM^\rT J \sM R_0^{-1/2} V,
\end{equation}
which, in view of (\ref{sMclos}),  depends in a quadratic fashion on the coupling shape matrices $\sM_1$, $\sL_1$, $\sM_2$, $\sL_2$. This can be used for a rational choice of the coupling shape in order to achieve given specifications on the stability margins and decoherence levels for the OQHO interconnection in the weak-coupling framework.

In the absence of direct coupling between the OQHOs, when $R_{12} = 0$ in (\ref{Hint}),  the matrix $R_0$ in (\ref{R0}) is block-diagonal:
\begin{equation}
\label{R0diag}
    R_0 = \blockdiag_{k=1,2}(R_k),
\end{equation}
so that the condition $R_0\succ 0$ in (\ref{typclos}) is equivalent to $R_k \succ 0$ for $k=1,2$. By combining (\ref{R0diag}) with  the block-diagonal structure of the CCR matrix $\Theta$ in (\ref{Theta12}), the matrix $V$ in (\ref{AVclos}), (\ref{VVclos}) takes the form
\begin{equation}
\label{Vdiag}
    V = \blockdiag_{k=1,2}(V_k),
\end{equation}
where $V_k \in \mC^{n_k\x n_k}$ is a unitary matrix whose columns are the eigenvectors of the Hermitian matrix
$$
    -2i\sqrt{R_k}\Theta_k \sqrt{R_k} = V_k\mho_k V_k^*,
    \qquad
    k = 1,2,
$$
with $\mho_k$ a diagonal matrix formed from the eigenfrequencies of $2\Theta_k R_k$. In this case,  the block-diagonal structure of the matrices $J$, $R_0$, $V$ in (\ref{Om12}), (\ref{R0diag}), (\ref{Vdiag}) alows the matrix $K$ in (\ref{K}) to be computed as
\begin{equation}
\label{mat}
    K
     =
    \begin{bmatrix}
        K_1 & *\\
        * & K_2
    \end{bmatrix},
    \quad
    K_j = -iV_j^* R_j^{-1/2}(\sM_j^\rT J_j \sM_j + \sL_j^\rT \wt{J}_{3-j} \sL_j)R_j^{-1/2}V_j,
    \qquad
    j = 1,2,
\end{equation}
where (\ref{tJk}), (\ref{sMclos}) are also used, with  the off-diagonal blocks ``$*$'' of the matrix $K$ in (\ref{mat}) being irrelevant for (\ref{muclos}).

\section{Conclusion}
\label{sec:conc}

For quantum harmonic oscillators, we have discussed their isolated and open dynamics, including the role of the system-field coupling matrix in the energy exchange with the environment and the decay in the two-point CCRs of the system variables. The decay time has been proposed as a decoherence measure for OQHOs  and equipped with an upper bound using algebraic Lyapunov inequalities. For the class of OQHOs with a nonsingular CCR matrix and a positive definite energy matrix, with different eigenfrequencies  in the absence of coupling,
we have applied spectrum perturbation techniques to  investigating the asymptotic behaviour of the leading Lyapunov exponent for the dynamics matrix of the OQHO with a small coupling strength parameter and a given coupling shape matrix. A threshold value has  been established for the coupling strength in order to guarantee that the oscillator retains the dynamic properties of its uncoupled version over a large number of cycles before the decay sets in. 
The asymptotic properties of the covariance matrix of the invariant Gaussian quantum state of the system in the  weak-coupling limit has also been discussed. 
These results have been demonstrated for one- and two-mode oscillators with multichannel external fields and also applied to a coherent feedback interconnection of OQHOs with direct and field-mediated coupling.  The findings of the paper can be used for the formulation of performance criteria for quantum feedback networks \cite{GJ_2009,JG_2010} requiring a  controlled isolation from the environment for quantum information processing applications \cite{NC_2000}. 
Another potential application of decoherence measures, discussed in the paper, is the robustness of coherent quantum control architectures \cite{JNP_2008,MP_2009,NJP_2009,SVP_2017,VP_2013a,ZJ_2011a} to unmodelled noises.

\end{document}